\documentclass[aop]{imsart}

\RequirePackage{amsthm,amsmath,amsfonts,amssymb}
\RequirePackage[numbers]{natbib}
\RequirePackage{mathrsfs}
\RequirePackage[colorlinks,citecolor=blue,urlcolor=blue]{hyperref}
\RequirePackage{graphicx}

\startlocaldefs
\theoremstyle{plain}
\newtheorem{theorem}{Theorem}[section]
\newtheorem{lemma}[theorem]{Lemma}
\newtheorem{proposition}[theorem]{Proposition}
\newtheorem{corollary}[theorem]{Corollary}

\theoremstyle{remark}
\newtheorem{definition}[theorem]{Definition}

\newtheorem{remark}[theorem]{Remark}

\newcommand{\build}[3]{\mathrel{\mathop{\kern 0pt#1}\limits_{#2}^{#3}}}

\def\SU{{\mathrm{SU}}}

\def\U{{\mathrm U}}
\def\Z{{\mathbb Z}}

\def\N{{\mathbb N}}

\def\R{\mathbb{R}}
\def\C{\mathbb{C}}
\def\Tr{\mathrm{Tr}}

\def\vol{\mathrm{vol}}
\def\Pbb{\mathbb{P}}
\def\Pfr{\mathfrak{Y}}
\def\E{\mathbb{E}}
\def\Gfr{\mathscr{G}}
\def\Ufr{\mathscr{U}}

\def\F{\mathcal{F}}
\def\chir{\mathrm{chir}}
\def\Ffr{\mathfrak{F}}
\def\PH{\mathfrak{M}}
\def\YM{\mathrm{YM}}
\def\Hom{\mathrm{Hom}}

\def\geq{\geqslant}
\def\leq{\leqslant}
\def\ge{\geq}
\def\le{\leq}

\endlocaldefs

\begin{document}

\begin{frontmatter}
\title{Gaussian measure on the dual of $\mathrm{U}(N)$, random partitions, and topological expansion of the partition function}
\runtitle{Gaussian measure on the dual of $\mathrm{U}(N)$}

\begin{aug}
\author[A]{\fnms{Thibaut}~\snm{Lemoine}\ead[label=e1]{thibaut.lemoine@college-de-france.fr}},
\and
\author[B]{\fnms{Myl\`ene}~\snm{Ma\"ida}\ead[label=e2]{mylene.maida@univ-lille.fr}}
\address[A]{Coll\`ege de France, 3, rue d'Ulm, 75005 Paris, France\printead[presep={,\ }]{e1}}

\address[B]{Universit\'e de Lille, CNRS, UMR 8524 -- Laboratoire Paul Painlev\'e, F-59000 Lille, France\printead[presep={,\ }]{e2}}
\end{aug}

\begin{abstract}
We study a Gaussian measure $\Gfr_N(q)$ with parameter $q\in(0,1)$ on the dual of the unitary group of size $N$: we prove that a random highest weight under $\Gfr_N(q)$ is the coupling of two independent $q$-uniform random partitions $\alpha,\beta$ and a random integer $n\sim\Gfr_1(q)$. We prove deviation inequalities for the $q$-uniform measure, and use them to show that the coupling of random partitions under $\Gfr_N(q)$ vanishes in the limit $N\to\infty$. We also prove that the partition function of this measure admits an asymptotic expansion in powers of $1/N$, and that this expansion is topological, in the sense that its coefficients are related to the enumeration of ramified coverings of elliptic curves. It provides a rigorous proof of the gauge/string duality for the Yang-Mills theory on a 2D torus with gauge group $\U(N),$ advocated by Gross and Taylor \cite{GT,GT2}.
\end{abstract}



\end{frontmatter}


\section{Introduction}

There is a longstanding study of random discrete models involving algebraic or combinatorial structures with  growing size, in particular distributions on Young diagrams. Multiple motivations hide behind these distributions: understanding the eigenvalue distribution of random matrices \cite{BOO}, estimating matrix integrals in quantum field theory \cite{NO,Oko3}, describing the representations of infinite unitary or symmetric groups \cite{BO05,BK08}, or even analysing the structure of integrable models arising from statistical physics \cite{Zyg}.

In the present paper, we are mainly interested in measures on the set $\widehat{\U}(N)$ of highest weights of the unitary group $\U(N)$. 
It corresponds to the set of nonincreasing $N$-tuples of integers $\lambda=(\lambda_1,\ldots,\lambda_N),$ labelling the irreducible representations of $\U(N).$  For any $q\in (0,1),$ we denote by $\Gfr_N(q)$ the distribution on $\widehat{\U}(N)$ given by
\[
\Pbb(\lambda = \mu) \propto q^{c_2(\mu) },\quad \forall \mu\in\widehat{\U}(N),
\]
where $c_2(\lambda)$ is the \emph{Casimir number} of the representation $\lambda$, given by 
\[
c_2(\lambda) = \frac1N \sum_{i=1}^N \lambda_i(\lambda_i+N+1-2i),
\]
which is (up to a sign) an eigenvalue of the Laplace operator on $\U(N)$.

Writing $q= e^{- \beta}$ for some inverse temperature $\beta >0$ makes explicit the fact that  $\Gfr_N(q)$ is a Gibbs measure of the form 
\[
\Pbb(\lambda=\mu) \propto e^{-\beta H(\mu)},\quad \forall \mu\in\widehat{\U}(N),
\]
where the Hamiltonian $H:\widehat{\U}(N)\to\R$ is the function that assigns to any highest weight its Casimir number. As $c_2$ is a quadratic form in the coordinates of $\lambda,$ it is even a (discrete) Gaussian measure. Gibbs measures of the same kind also appear in the character expansion of some unitary matrix models \cite{GM}.

The main quantity that we will investigate is the \emph{partition function} of the model, which is given by
\[
Z_N(q) = \sum_{\lambda\in\widehat{\U}(N)} q^{c_2(\lambda)}.
\]

In the case of $\widehat{\U}(1)=\Z,$ the partition function is nothing but the Jacobi theta function
\[
\theta(q) = \sum_{n\in\Z} q^{n^2},
\]
which is analytic on the unit disk $\{q\in\C:\vert q\vert <1\}$. We will see that the measure $\Gfr_N(q)$ actually corresponds to a coupling of two random partitions by means of a $\Gfr_1(q)$-distributed integer.

A natural question, addressed in many earlier works for measures on integer partitions -- see \cite{LS,VK,Ver,BOO,Oko,Rom2,CLW} and references therein -- is to understand the limiting distribution of the model. We want to treat a similar question, but for highest weights instead of partitions.

\subsection{Coupling of random partitions}

The highest weights of $\U(N)$ can be related to triples $(\alpha,\beta,n)$ where $\alpha,\beta$ are integer partitions and $n$ is an integer. We denote by $\Pfr$  the set of all integer partitions and for any $\alpha=(\alpha_1\geq\cdots\geq\alpha_r>0) \in \Pfr$, we denote by $\ell(\alpha)=r$ its number of parts. As pointed out in the literature \cite{Lit1,Lit2,Sta,Koi}, given two partitions $\alpha$ and $\beta$ and an integer $n$, for any $N\geq\ell(\alpha)+\ell(\beta)$, one can build a highest weight $\lambda_N(\alpha,\beta,n)$ of $\U(N)$ according to scheme described in Figure \ref{fig:decomp_highest_weight}.

\begin{figure}[b]
    \centering
    \includegraphics{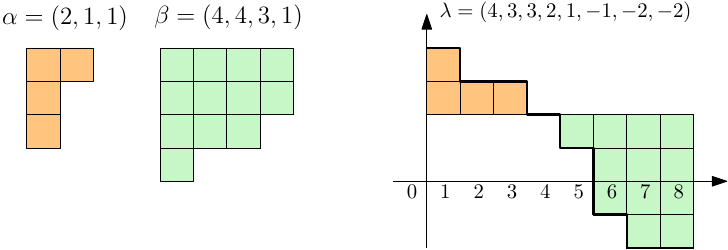}
    \caption{Construction of the highest weight $\lambda_N(\alpha,\beta,n)$ for $N=8$, $\alpha=(2,1,1)$, $\beta=(4,4,3,1)$ and $n=2$.}
    \label{fig:decomp_highest_weight}
\end{figure}

It turns out that, for any $N \ge 1,$ if we add the restriction that the triple $(\alpha,\beta,n)$ should belong to 
\[
\Lambda_N = \{(\alpha,\beta,n)\in\Pfr^2\times\Z:\ell(\alpha)\leq A_N,\ell(\beta)\leq B_N\},
\]
with 
\begin{equation} \label{eq:ANBN}
A_N := \lfloor (N+1)/2\rfloor-1 \textrm{ and } B_N :=  N-\lfloor(N+1)/2\rfloor, 
\end{equation}
then, there is a converse construction, namely a bijection $\Phi_N:\widehat{\U}(N)\to\Lambda_N$. More details will be given in Section \ref{sec:Gfr} (see in particular Proposition \ref{prop:bijection_hw_partitions}). We will express the pushforward $(\Phi_N)_*\Gfr_N(q)$ in terms of $\Gfr_1(q)$ and a $q$-deformation of the uniform measure on the set of partitions.

A random partition $\alpha \in\Pfr$ is $q$-uniform if
\begin{equation}
\Pbb(\alpha=\beta) \propto q^{\vert\beta\vert},\quad \forall \beta\in\Pfr,
\end{equation}
and we denote by $\Ufr(q)$ the $q$-uniform distribution on $\Pfr$. We will show that for any $N\geq 1$ and any $q\in(0,1)$, if $\lambda \sim \Gfr_N(q),$ there exists a function $F_{N,q}:\Lambda_N\to\R$ such that, for all bounded measurable $f:\Pfr^2\times\Z\to\R$,
\[
\E[f\circ\Phi_N(\lambda)] = \E[f(\alpha,\beta,n) F_{N,q}(\alpha,\beta,n)],
\]
where $\lambda\sim\Gfr_N(q)$, $\alpha,\beta$ are two independent random partitions with distribution $\Ufr(q)$ and $n$ is a random integer with distribution $\Gfr_1(q)$, independent from $\alpha$ and $\beta$. The precise result is stated in Proposition \ref{prop:transfert}. A remarkable fact is that the coupling between $\alpha, \beta,$ and $n,$ quantified through $F_{N,q},$ vanishes asymptotically as $N\to\infty:$

\begin{theorem}\label{thm:cv_distrib}
For any $q\in(0,1)$, we have the following convergence in distribution:
\begin{equation}
(\Phi_N)_*\Gfr_N(q) \build{\longrightarrow}{N\to\infty}{\mathcal{L}} \Ufr(q)^{\otimes 2}\otimes\Gfr_1(q).
\end{equation}
\end{theorem}
In other terms, if $\alpha,\beta,n$ are such that $\lambda_N(\alpha,\beta,n)$ has distribution $\Gfr_N(q)$, then as $N$ tends to infinity they become independent, with $\alpha,\beta\sim\Ufr(q)$ and $n\sim\Gfr_1(q)$.

\subsection{Asymptotic expansion of the partition function}

Building on the formulation of highest weights as couplings of partitions, on  their asymptotic independence and on deviation inequalities under $\Ufr(q)$ that we establish in Section \ref{sec:deviation}, we will provide an asymptotic expansion of $Z_N(q)$ in powers of $\frac{1}{N}$ up to any order. By symmetry, coefficients of odd powers exactly vanish, leading to an expansion in powers of $\frac{1}{N^2}:$
\begin{theorem}\label{thm:main}
For any $q\in(0,1),$  the partition function $Z_N(q)$ admits the following expansion: for any $p\geq 1$,
\begin{equation}
Z_N(q) = \sum_{k=0}^p \frac{a_{2k}(q)}{N^{2k}} + O(N^{-2p-2}).
\end{equation}
\end{theorem}

All coefficients are explicit and  will be provided in a more detailed way in Theorem \ref{thm:main3}. We will now make the link with two-dimensional Yang--Mills theory and show that, up to a simple correspondence of parameters, the partition function $Z_N(q)$ is the partition function of the latter on a torus.

\subsection{Two-dimensional Yang--Mills theory}

The two-dimensional Yang--Mills theory with gauge group $\U(N)$ in the large $N$ limit has been popularized by the breakthrough work from 't Hooft \cite{Hoo74} to describe strong interaction in the large coupling regime. The probabilistic counterpart of this quantum field theory is the Yang--Mills holonomy process. It is a stochastic process $(H_\ell)$ indexed by a space of loops on a compact connected orientable surface $\Sigma_{g,t}$ of genus $g$ and area $t,$ with values in a compact Lie group $G$. It has been constructed by Sengupta \cite{Sen} and L\'evy \cite{Lev00} based on the fact that its finite-dimensional marginals correspond to the discrete Yang--Mills measure on an area-preserving cellular embedding of a graph in $\Sigma_{g,t}$. A recent overview of this process and its relationship with the quantum Yang--Mills theory considered by physicists can be found in the lecture notes by L\'evy \cite{Lev2}.  
In the case when $G=\U(N)$ and if its Lie algebra is endowed with the inner product $\langle X,Y\rangle = N\Tr(XY^*)$, using the character expansion of the heat kernel on $\U(N),$ the partition function of the model can be expressed as
\begin{equation} \label{eq:Z-YM}
Z_N^\YM(g,t) = \sum_{\lambda\in\widehat{\U}(N)}e^{-\frac{t}{2}c_2(\lambda)} d_\lambda^{2-2g},
\end{equation}
where $d_\lambda$ denotes the dimension of the irreducible representation of $\U(N)$ with highest weight $\lambda$. In the case  $g=1$, we have  
\[
Z_N^\YM(1,t)= Z_N(e^{-\frac{t}{2}}).
\]
Therefore, the  asymptotic expansion given by Theorem \ref{thm:main} can be immediately translated into an asymptotic expansion for the partition function $Z_N^\YM(1,t)$ on a torus.

As a direct corollary of our asymptotic expansion, we recover the convergence, for any value of $t$ of the Yang--Mills partition function on a torus to a finite value, which was proved by the first author through other means. Let us define the Euler function $\phi$ as follows: for any $q\in\C$ such that $\vert q\vert<1,$ 
\begin{equation}
\label{eq:phi}    
\phi(q) = \prod_{m=1}^\infty (1-q^m).
\end{equation}

\begin{theorem}[\cite{Lem}]\label{thm:lim_pf}
For any $t>0$, setting $q_t=e^{-\frac{t}{2}}$, we have
\begin{equation}\label{eq:lim_pf_TL}
\lim_{N\to\infty} Z_N^\YM(1,t) = \frac{\theta(q_t)}{\phi(q_t)^{2}}.
\end{equation}
\end{theorem}

The study of the partition function in higher genus needs an additional analysis, coming from the contribution of negative powers of dimensions of representations. The partition function in this case can be seen as a $q$-deformation of the Witten zeta function
\[
\zeta_{\SU(N)}(2g-2)=\sum_{\lambda\in\widehat{\SU}(N)}d_\lambda^{2-2g}.
\]
An asymptotic expansion of this function has been proved by Magee \cite{Mag}, without any topological interpretation; in a forthcoming paper, we will address the asymptotic expansion of the partition function for $g\geq 2$ based on the results of \cite{Mag} and the present paper. We also expect that the tools developed for the partition function might work for the general distribution of the Yang--Mills holonomy field, leading to an improvement of some recent results obtained by the first author with  Dahlqvist \cite{DL,DL2}.

\subsection{On the notion of topological expansion}\label{subsec:topo}

In theoretical physics, in particular in the so-called \emph{matrix models}, there is a formal counterpart of the matrix integrals called the \emph{formal matrix integrals}, consisting of the sums of power series expansions which are not required to converge \cite{Eyn2}. These power series expansions are often called \emph{topological expansions} when their coefficients are related to topological invariants, or to the enumeration of coverings or topological maps. The idea is that, whether the power series converges or not, computing these coefficients become more tractable because of this topological/geometric interpretation. As discussed by Eynard \cite{Eyn2,Eyn}, there has been for a long time a confusion between asymptotic expansions of matrix integrals and topological expansions of formal matrix integrals, and many results from the theoretical physics literature, for instance the seminal papers \cite{Hoo74,BIPZ}, are rigorous at the level of \emph{formal matrix integrals}, which is sufficient from a combinatorial perspective. The remaining question is whether these formal matrix integrals coincide with the convergent matrix integrals; it has been answered in a few special cases, see for instance \cite{BI,GZ,GM,GMS,GN}. Concerning the Yang--Mills partition function $Z_N^\YM(g,t)$, it is known to admit a formal series expansion by major series of works in theoretical physics \cite{GT,GT2,BT,Ram,CMR}. Our claim is that the asymptotic expansion in Theorem \ref{thm:main} matches the formal power series expansions obtained by Gross and Taylor \cite{GT,GT2}. It is not only a nontrivial result, but also contradicts the predictions of Gross and Taylor, who stated that the topological expansion that they exhibited diverges. The precise statement of this topological expansion is given in Theorem \ref{thm:top_exp_torus}. In a slightly simplified setting, one can consider the so-called \emph{chiral partition function} obtained by isolating only one of the coupled partitions $(\alpha,\beta)$
\[
Z_N^\chir(1,t) = \sum_{\substack{\alpha\in\Pfr\\ \ell(\alpha)\leq N}}e^{-\frac{t}{2}c_2(\alpha)}.
\]
In this case, we are able to prove  the following.

\begin{theorem}\label{thm:chiral_pf}
The chiral partition function of the Yang--Mills measure on a torus of area $t>0$ admits the following asymptotic expansion, for any $p\in\N$:
\begin{equation}\label{eq:chiral_pf}
Z_N^\chir(1,t) = \sum_{k=0}^{p}\frac{t^{2k}}{(2k)!N^{2k}}\F_{1,k+1}(e^{-\frac{t}{2}}) + O(N^{-2p-2}),
\end{equation}
where $\F_{1,k+1}$ denotes the generating function of the number of (generically) ramified coverings of the torus with genus $k+1$.
\end{theorem}

If we denote by $\rho_t$ a properly defined uniform measure on the space $\mathcal{R}$ of the ramified coverings of the torus of area $t$, Equation \eqref{eq:chiral_pf} can be turned into an expression in terms of \emph{random ramified coverings}, as conjectured in \cite{GT}:
\begin{equation}
Z_N^\chir(1,t) = \int_{\mathcal{R}}e^{-\frac{t}{2}\deg(R)}N^{\chi(R)}\rho_t(dR),
\end{equation}
where $\deg(R)$ and $\chi(R)$ are respectively the \emph{degree} and \emph{Euler characteristic} of the ramified covering $R$. The precise statement is detailed in Section~\ref{sec:expansions}. Such integrals appear in string theory, where the ramified coverings of the torus are seen as worldsheets, i.e. string embeddings, and the interplay between Yang--Mills theory and string theory is the \emph{gauge/string duality} developed (at least at the level of formal power series) by Gross and Taylor. The only other proofs of such duality, for convergent matrix integrals, seem to have been found in the plane \cite{Lev,PPSY}.

\subsection{Organisation of the paper}

In Section \ref{sec:phi_theta}, we provide  a systematic study of the measure $\Ufr(q)$ on $\Pfr$, and we give deviation inequalities for this measure that will be used in the remainder of the paper. 
In  Section \ref{sec:Gfr}, we introduce the measure 
$\Gfr_N(q).$ After a quick description of $\Gfr_1(q),$
we will provide the relationship between $\Ufr(q)$, $\Gfr_1(q)$ and $\Gfr_N(q).$ A key  feature is the construction of an explicit bijection between the set of highest weights and the subset $\Lambda_N\subset\Pfr^2\times\Z.$ In Section \ref{sec:largeN}, we will prove the asymptotic decoupling described in Theorem \ref{thm:cv_distrib} and the asymptotic expansion of Theorem \ref{thm:main}.
In Section \ref{sec:expansions}, we will show that in the case of the torus, the asymptotic expansion provided by Theorem \ref{thm:main} is in fact a topological expansion and we will also discuss the chiral case. Preliminary material for topological expansions will be gathered in Section \ref{sec:enumeration}, in particular we will define and study the Frobenius measures in relation with the enumeration of ramified coverings.\\

{\bf Acknowledgements.} Both authors would like to thank Justine Louis for useful discussions at the beginning of the project. T.L. would also thank Antoine Dahlqvist for pointing him out the reference \cite{Ver} at an early stage of this work. This work is partially supported by ANR AI chair BACCARAT (ANR-20-CHIA-0002) and by the Labex CEMPI (ANR-11-LABX-0007-01).

\section{$q$-Uniform measure on partitions}\label{sec:phi_theta}

In this section we study in detail the measure $\Ufr(q)$ and prove deviation inequalities needed for the asymptotic expansions.

\subsection{Definitions and main properties}\label{sec:def_prop_uq}

Before going into the study of  the measure $\Ufr(q),$ we recall some well-known facts
about integer partitions. An \emph{integer partition}\footnote{It is more usual to denote integer partitions by $\lambda$, but we rather keep the notations $\lambda,\mu$ for highest weights and use $\alpha,\beta$ for partitions, in order to avoid confusions, apart from specific cases where a highest weight is also a partition.} $\alpha\vdash n$ of size $n$ is a finite family of integers $\alpha=(\alpha_1,\ldots,\alpha_{\ell(\alpha)})\in\N^{\ell(\alpha)}$ with $\alpha_1\geq\cdots\geq\alpha_{\ell(\alpha)}>0$, $|\alpha|:=\alpha_1+\cdots +\alpha_{\ell(\alpha)} = n$. The integers $\alpha_1,\ldots,\alpha_{\ell(\alpha)}$ are called the \emph{parts} of the partition and $\ell(\alpha)$ is called its \emph{length}. We denote by $\Pfr_n$ the set of partitions of size $n$, and $\Pfr=\bigsqcup_n \Pfr_n$ the set of all integer partitions. A partition $\alpha\vdash n$ is usually represented by a \emph{Young diagram}, which is a collection of $n$ boxes in $\ell(\alpha)$ left-justified rows, where the $i-$th row contains $\alpha_i$ boxes. An example is displayed in Figure \ref{fig:young_diagram}.
\begin{figure}[t]
    \centering
    \includegraphics{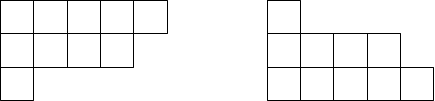}
    \caption{The Young diagram of the partition $(5,4,1)$, in the English convention (on the left), and the French convention (on the right).}
    \label{fig:young_diagram}
\end{figure}
In this representation, each box is labelled $(i,j)$, where $1\leq i\leq\ell(\alpha)$ (resp. $1\leq j\leq \alpha_i$) is the number of the row (resp. column) that contains the box. The \emph{content} of the box $\square=(i,j)$ is defined as
\[
c(\square)=j-i.
\]
The \emph{total content} of a partition $\alpha$ is the sum of the contents of each boxes in the corresponding Young diagram:
\begin{equation}
\label{def:content}
K(\alpha) = \sum_{\square\in\alpha} c(\square).
\end{equation}

It is well known that integer partitions of size $n$ are in bijection with the irreducible representations of the symmetric group $\mathfrak{S}_n$ (see e.g.  \cite{Sag} for a nice introduction).  Consequently, for each partition $\alpha\vdash n,$ one can define  the character of the irreducible representation of $\mathfrak{S}_n$ as the trace of the corresponding representation, and we shall denote it by $\chi_\alpha.$ 

Let us now define our family of measures $\Ufr(q).$

\begin{definition}
Let $q\in(0,1)$ be a real parameter. A random partition $\alpha\in\Pfr$ is distributed according to the \emph{$q$-uniform measure} $\Ufr(q)$ if
\begin{equation}
\Pbb(\alpha=\beta)=\phi(q)q^{\vert \beta\vert},\quad \forall \beta\in\Pfr,
\end{equation}
where $\phi$ is the Euler function defined in \eqref{eq:phi}.
\end{definition}\label{eq:euler}
The fact that $\Ufr(q)$ is a probability measure comes from the classical equality
\begin{equation}
\sum_{\alpha\in\Pfr} q^{\vert\alpha\vert} = \sum_{n\geq 0}p(n)q^n = \phi(q)^{-1},\ \forall q\in\C, \ \vert q\vert <1,
\end{equation}
with $p(n)$ the number of partitions of the integer $n.$ The name  \textit{$q$-uniform} comes from the fact that a $q$-uniform random partition $\alpha$, conditionnally on the event $\{\alpha\vdash n\}$, is actually uniformly distributed in $\Pfr_n$. 

\begin{proposition}\label{prop:phi_condition}
Let $q\in(0,1)$ be fixed, and $\alpha\in\Pfr$ be a random partition with distribution $\Ufr(q)$. The distribution of $\vert\alpha\vert$ is given by
\begin{equation}
\Pbb(\vert\alpha\vert = n) = \phi(q) p(n)q^n,\ \forall n\in\N.
\end{equation}
Moreover, conditionally to $\vert\alpha\vert=n$, $\alpha$ is uniform on $\Pfr_n$.
\end{proposition}

\begin{proof}
The expression of $\Pbb(\vert\alpha\vert = n)$, for $n\in\N$, follows from the definition of $p(n).$ Then, we have for any $\beta\in\Pfr$ and any $n\in\N$,
\[
\Pbb(\alpha = \beta \ \vert\  \vert\alpha\vert = n) = \frac{\Pbb(\{\alpha=\beta\} \cap \{\vert\alpha\vert = n\})}{\Pbb(\vert\alpha\vert = n)} = \frac{1}{p(n)},
\]
which corresponds indeed to the uniform measure on $\Pfr_n$.
\end{proof}

This property of the $q$-uniform measure described in Proposition \ref{prop:phi_condition} is reminiscent from the Poissonized Plancherel measure: recall that the distribution of $\alpha$ is the Plancherel measure on $\Pfr_n$ when
\[
\Pbb(\alpha = \beta) = \frac{(f^\beta)^2}{n!},\ \forall \beta\vdash n,
\]
where $f^\beta$ denotes the number of standard Young tableaux of shape $\beta$, or equivalently the dimension of the irreducible representation of $\mathfrak{S}_n$ associated to $\beta$ (see \cite{Sag} for instance). Analogously, $\alpha$ is distributed according to the Poissonized Plancherel measure on $\Pfr$ of parameter $t$ when
\[
\Pbb(\alpha = \beta) = e^{-t} t^{\vert\beta\vert}\frac{(f^\beta)^2}{\vert\beta\vert!}.
\]
In this case the distribution of $\vert\alpha\vert$ is the Poisson measure of parameter $t$ and conditionally to $\vert\alpha\vert = n$, $\alpha$ follows the Plancherel distribution on $\Pfr_n$. Another common point between the Poissonized Plancherel and $q$-uniform measures is that in both cases, the random point process $\{\alpha_i+1/2-i\}_{i \in \mathbb Z}$ is determinantal \cite{BOO,BO}. Moreover, as observed by Borodin \cite{Bor}, there is a deeper relationship between the Poissonized Plancherel measure and $\Ufr(q)$ in the sense that both measures are special limits of a more general family of measures introduced by Nekrasov and Okounkov in the study of Seiberg--Witten theory \cite{NO}. We will introduce another family of measures in Section \ref{sec:enumeration}, coming from enumeration of surface coverings, that contain the $q$-uniform and Poissonized Plancherel measures as particular cases.

\subsection{Deviation  inequalities}\label{sec:deviation}

We require two estimates related to the generating function of $p(n)$. For any integers $n,k$ we set
\[
p_k(n)=\#\{\alpha\vdash n:\ell(\alpha)=k\},
\]
as well as
\[
p_{\leq k}(n) = p_1(n)+p_2(n)+\cdots+p_k(n).
\]
It is clear that, for $n$ fixed, the sequence $(p_{\leq k}(n))_{k\geq 1}$ is increasing, bounded, and is stationary equal to $p(n)$ for $k\geq n$. We start with the following lemma:

\begin{lemma}\label{lem:remainder_pn}
For any $M\geq 1$ and any $q\in(0,1)$, we have the following inequality:
\begin{equation}
\sum_{n\in\N} (p(n)-p_{\leq M}(n))q^n \leq 5.44 q^{M+1} \sum_{m=0}^\infty (1+ e^{\pi\sqrt{\frac{2m}{3}}}) e^{m\log q}.
\end{equation}
\end{lemma}

\begin{proof} 
For any $M\geq n$, $p_{\leq M}(n)=p(n)$, we have
\[
\sum_{n\in\N} (p(n)-p_{\leq M}(n))q^n = \sum_{n\geq M+1} (p(n)-p_{\leq M}(n))q^n.
\]
Then, we know from \cite[Corollary 1]{Oru} that for any $1\leq k\leq n-1$,
\[
p_k(n)\leq \frac{5.44}{n-k}e^{\pi\sqrt{\frac{2(n-k)}{3}}}.
\]
It follows that for any $n\geq M+1$,
\[
\sum_{k=M+1}^{n-1}p_k(n) \leq \sum_{k=M+1}^{n-1} f(n-k),
\]
where $f(x) = \frac{5.44}{x}e^{\pi\sqrt{\frac{2x}{3}}}$. By a change of variable $\ell=n-k$, it reads
\[
\sum_{k=M+1}^{n-1}p_k(n) \leq \sum_{\ell = 1}^{n-M-1} f(\ell).
\]
However, as $f$ is increasing on $[1,\infty)$, we can bound each term from above by $f(n-M-1)$. We finally obtain
\[
\sum_{k=M+1}^{n-1}p_k(n)\leq (n-M-1)f(n-M-1).
\]
As a conclusion,
\[
p(n)-p_{\leq M}(n) = p_n(n)+ \sum_{k=M+1}^{n-1}p_k(n)\leq 1 + (n-M-1)f(n-M-1).
\]
We can put it back in the sum over $n$, yielding
\[
\sum_{n\in\N} (p(n)-p_{\leq M}(n))q^n \leq \sum_{n\geq M+1} (1+(n-M-1)f(n-M-1))q^n.
\]
We can perform a change of variable $m = n-M-1$:
\[
\sum_{n\in\N} (p(n)-p_{\leq M}(n))q^n \leq \sum_{m\geq 0} (1+ m f(m))q^{m+M+1} = q^{M+1}\sum_{m\geq 0} (1+m f(m))q^m.
\]
\end{proof}

Note that the sum in the previous lemma is  finite for $|q|<1$ by the integral test for convergence. In the following lemma, we make clear that it decays exponentially in terms of $M$.

\begin{lemma}\label{lem:remainder_pn2} 
For any $q\in(0,1)$, there exists a constant $C_q$ such that, for any $M \ge 1,$ 
\begin{equation}
\sum_{n=M+1}^\infty p(n)q^n \leq C_q\cdot e^{\frac{1}{2} M\log q}.
\end{equation}
\end{lemma}

\begin{proof}
We use the standard upper bound
\[
p(n) < e^{c\sqrt{n}},\ \forall n\geq 1,
\]
with $c=\pi\sqrt{2/3}$ (see \cite[Theorem 14.5]{Apo} for instance). It yields
\[
\sum_{n=M+1}^\infty p(n)q^n \leq \sum_{n=M+1}^\infty e^{c\sqrt{n}+n\log q}.
\]
By a direct change of index, we have
\[
\sum_{n=M+1}^\infty p(n)q^n \leq \sum_{n=1}^\infty e^{c\sqrt{n+M}+(n+M)\log q}.
\]

As $\sqrt{n+M} \le \sqrt n + \sqrt M,$ we have that 
\[
\sum_{n=M+1}^\infty p(n)q^n \leq e^{M\log q + c\sqrt M} \sum_{n \ge 1} e^{n\log q + c\sqrt n} 
\]
We denote by $S_q := \sum_{n \ge 1} e^{n\log q + c\sqrt n},$ the sum being finite when $|q|<1.$
We conclude by observing that there exists $T_q$ such that for any $M\ge 1,$
\[ e^{M\log q + c\sqrt M} \le  T_q e^{\frac{1}{2}M\log q }, \]
and the lemma holds with $C_q = S_qT_q.$
\end{proof}

As a direct consequence, let us start with some simple deviation inequalities for partitions with distribution $\Ufr(q)$:

\begin{proposition}\label{prop:tail_pt}
Let $q\in(0,1)$ be fixed, and $\alpha\sim\Ufr(q)$ be a random partition. 
Then there exists a constant $C_q$ such that 
\[ 
\Pbb(\ell(\alpha) > \lfloor(N-1)/2\rfloor) \leq C_q e^{\frac{N}{2} \log q}
\]
and 
\[ 
\Pbb(\vert\alpha\vert > N) \leq C_q e^{\frac{N}{2} \log q}.
\]

Consequently, for any $p\in\N$, there are constants $C_1(p,q),C_2(p,q)>0$ such that
\begin{equation}
\Pbb(\ell(\alpha) > \lfloor(N-1)/2\rfloor) \leq C_1(p,q) N^{-p},
\end{equation}
and
\begin{equation}
\Pbb(\vert\alpha\vert > N) \leq C_2(p,q) N^{-p}.
\end{equation}
\end{proposition}

\begin{proof}
We start by writing the expression of the probability that $2\ell(\alpha)>N$.
\[
\Pbb(2\ell(\alpha) > N) = \phi(q)\sum_{\substack{\alpha\in\Pfr\\ 2\ell(\alpha)>N}} q^{\vert\alpha\vert} = \phi(q) \sum_{n\geq 0} q^n \sum_{\substack{\alpha\vdash n\\ 2\ell(\alpha)>N}}1.
\]
The sum over $\alpha$ is nothing else than $p(n)-p_{\leq \lfloor (N-1)/2\rfloor}(n)$. We can apply Lemma \ref{lem:remainder_pn} to get that 
\[
\Pbb(2\ell(\alpha) > N)  \le \phi(q) 5.44 q^{\lfloor (N-1)/2\rfloor +1}
\sum_{m \ge 0} (1+e^{\pi \sqrt{\frac{2m}{3}}})e^{m \log q}.
\]
As the sum is finite and $\lfloor (N-1)/2\rfloor +1 \ge \frac{N}{2},$ we get the first inequality.

The second one is very similar :
\[
\Pbb( \vert\alpha\vert > N) \le \phi(q) \sum_{n \ge N+1} q^n p(n)
\]
and we proceed similarly, by using Lemma \ref{lem:remainder_pn2} instead of Lemma \ref{lem:remainder_pn}.
\end{proof}

\section{The Gaussian measure on $\widehat{\U}(N)$}\label{sec:Gfr}

In this section, we will study the measure $\Gfr_N(q)$ on two different aspects. First, we shall describe the specific case of $N=1$ where all moments can be computed explicitly, then we will prove that in the general case a highest weight of law $\Gfr_N(q)$ is a coupling of two random $q$-uniform partitions and a highest weight of $\U(1)$ with law $\Gfr_1(q)$. We will also prove Theorem \ref{thm:cv_distrib} which can be interpreted as an asymptotic decoupling.

\subsection{Definition of $\Gfr_N(q)$}

 One of the most important results in the representation theory of compact Lie groups is the celebrated Peter--Weyl theorem, which states that the characters of irreducible representations of such groups form a Hilbert basis of the space of square integrable central functions on $G$. Moreover, if we denote by $\Delta_G$ the Laplace--Beltrami operator on $G$ and by $\widehat{G}$ the set of equivalence classes of irreducible representations of $G$ (called \emph{dual} of $G$), we have the following:

\begin{theorem}
Let $G$ be a compact Lie group. The set of irreducible characters $\{\chi_\lambda,\lambda\in\widehat{G}\}$ is a Hilbert basis of eigenfunctions of $\Delta_G$, and there is a family of nonnegative numbers,  $(c_2(\lambda),\lambda\in\widehat{G}),$ called the (quadratic) Casimir numbers, such that
\[
\Delta_G\chi_\lambda = -c_2(\lambda)\chi_\lambda,\ \forall \lambda\in\widehat{G}.
\]
\end{theorem}

In the case of $G=\U(N)$, the irreducible representations can be labelled by tuples of integers $\lambda=(\lambda_1,\ldots,\lambda_N)\in\Z^N$ such that $\lambda_1\geq\cdots\geq\lambda_N$, called \emph{highest weights}. 
If we equip the Lie algebra $\mathfrak{u}(N)$ with the metric given by the Killing form
\[
\langle X,Y\rangle = N\Tr(XY^*),\ \forall X,Y\in\mathfrak{u}(N),
\]
then the Casimir numbers are given as follows:
\begin{equation}
c_2(\lambda) = \frac{1}{N}\sum_{i=1}^N\lambda_i(\lambda_i+N+1-2i),\  \forall \lambda\in\widehat{\U}(N).
\end{equation}
We will see that it is a natural scaling for asymptotic analysis.

As explained in the introduction, we are interested in the measure $\Gfr_N(q)$ defined as follows, for $q\in(0,1)$: a random highest weight $\lambda\in\widehat{\U}(N)$ is distributed according to $\Gfr_N(q)$ if
\begin{equation} \label{def:GNQ}
\Pbb(\lambda=\mu) = \frac{1}{Z_N(q)}q^{c_2(\mu)},\ \forall \mu\in\widehat{\U}(N).
\end{equation}
If we set $q_t=e^{-\frac{t}{2}},$ we see that $\Gfr_N(q_t)$ is actually a discrete Gaussian measure on $\widehat{\U}(N).$ Indeed, seeing the set of highest weights as a discrete subset of the Weyl chamber
\[
\mathcal{C}_N=\{(x_1,\ldots,x_N)\in\R^N: x_1\geq\cdots\geq x_N\},
\]
the measure $\Gfr_N(q)$ can be seen as a measure on $\mathcal{C}_N$ with discrete support. If we denote by $\rho$ the element of $\mathcal{C}_N$ such that, for any $1 \leq i \leq N,$
$\rho_i = i - \frac{N+1}{2},$ then 
\[
\Pbb(\lambda=\mu)  \propto e^{-\frac{t}{2}\Vert\mu-\rho\Vert^2},
\]
hence it is indeed a kind of \emph{discrete Gaussian measure} in the Weyl chamber.

\subsection{The case $N=1$}

The measure $\Gfr_1(q)$, as explained in the next section, is involved in the study of the general measure $\Gfr_N(q)$ by a coupling with random partitions. The highest weights of $\U(1)$ are just integers, and for any $n\in\Z,$ its Casimir number is equal to $n^2$. It means that $n\sim\Gfr_1(q)$ if and only if
\begin{equation}
\Pbb(n=m) = \frac{1}{Z_1(q)}q^{m^2}, \ \forall m\in \Z.
\end{equation}
The partition function is the Jacobi theta function
\[
Z_1(q)=\theta(q):=\sum_{n\in\Z} q^{n^2},
\]
which is analytic for $\vert q\vert <1$. The measure $\Gfr_1(q)$ has been used by Borodin \cite{Bor} as a way to modify the periodic Schur process in order to obtain a determinantal point process. We will only consider the case where $q\in(0,1)$ in order to keep positive measures. The moments and cumulants of this distribution have already been described by Romik \cite{Rom} in the case $q=e^{-\pi}$ and by Wakhare and Vignat \cite{WV} in the general case, yielding deep formulas in terms of elliptic integrals, and the reader interested in the interplay between probabilistic and number-theoretic aspects of the theta function can have an eye on \cite{BPY,SV}. We shall provide a simple alternative expression of the moments obtained by Wakhare and Vignat. It involves the differential operators $D^n=\left(x\frac{d}{dx}\right)^n$ for all $n\geq 0$, which are related to the \emph{infinite wedge representation} of the Lie algebra $\mathcal{D}$ of differential operators on the unit circle \cite{BO}.

\begin{proposition}\label{prop:is_it_wick}
For any $q\in(0,1)$ and $k\geq 0$, if $n\sim\Gfr_1(q)$ we have
\begin{equation}\label{eq:theta_moment_even}
\E[n^{2k}] = \frac{D^{k}\theta(q)}{\theta(q)}, \ \text{and} \ \ \E[n^{2k+1}] = 0.
\end{equation}
\end{proposition}

It would be satisfactory to interpret Proposition \ref{prop:is_it_wick} as a discrete analog of the Wick formula for the Gaussian distribution, and we conjecture that its even moments are in fact counting functions of combinatorial objects. Surprisingly, we will find similar expressions of moments of another measure in Proposition \ref{prop:moments_frob}, where the theta function is replaced by some generating series.

We now go to the study of the measure $\Gfr_N(q)$ in the large $N$ limit.

\subsection{The general case: from highest weights to partitions and back}\label{sec:hw_partitions}

The main issue with the current expression  \eqref{def:GNQ} of $\Gfr_N(q)$ is that not only the summand depends on $N$, but also does the set $\widehat{\U}(N)$, which is not convenient for letting $N\to\infty$. A better way to deal with this issue is to rewrite $\widehat{\U}(N)$ in terms of partitions, using a decomposition of highest weights that has been widely used in the last decades \cite{Ste,Sta,Koi,GT,GT2,Lem,Mag,Lem3}, but seems to be already known since the works of Littlewood \cite{Lit2, Lit1}.  We recall the following result, that already appears in \cite{Lem} in a different form.

\begin{proposition}\label{prop:bijection_hw_partitions}
For any $N\geq 1$, there is a bijection
\[
\Phi_N:\widehat{\U}(N)\build{\longrightarrow}{}{\sim}\Lambda_N,
\]
where 
\[
\Lambda_N = \{(\alpha,\beta,n)\in\Pfr^2\times\Z:\ell(\alpha)\leq A_N,\ell(\beta)\leq B_N\}.
\]
has been defined in \eqref{eq:ANBN}.
\end{proposition}

\begin{proof}
The bijection consists in fixing $n=\lambda_{\lfloor (N+1)/2\rfloor}\in\Z$ and seeing the left part of $\lambda$ as $\alpha+n=(\alpha_1+n,\ldots,\alpha_r+n)$ and the right part as $n-\beta,$ as shown in Figure \ref{fig:decomp_highest_weight} from the introduction. We shall describe this bijection in more details: let $k=\lfloor (N+1)/2\rfloor$ and for any $\lambda=(\lambda_1,\ldots,\lambda_N)\in\widehat{\U}(N)$  denote by $\ell_1$ the greatest integer such that $\lambda_{\ell_1}>\lambda_k$ and $\ell_2$ the smallest integer such that $\lambda_{\ell_2}<\lambda_k$\footnote{It is important that $\ell(\alpha)<k$ and $\ell(\beta)\le N-k$, otherwise there would be several possible decompositions, by shifting $n$ and either $\alpha$ or $\beta$.}. Then, we set $\alpha_i=\lambda_i-\lambda_k$ for $1\leq i\leq\ell_1$ and $\beta_j = \lambda_k -\lambda_{N-j-1},$ for any $1 \le j \le \ell_2.$ We obtain
\begin{equation}\label{eq:lambda_odd}
\lambda = (\alpha_1+\lambda_k,\ldots,\alpha_{\ell_1}+\lambda_k,\underbrace{\lambda_k,\ldots,\lambda_k}_{N-\ell_1-\ell_2},\lambda_k-\beta_{\ell_2},\ldots,\lambda_k-\beta_1),
\end{equation}
where $\alpha=(\alpha_1,\ldots,\alpha_{\ell_1})$ and $\beta=(\beta_1,\ldots,\beta_{\ell_2})$ are indeed integer partitions. Furthermore, we also have by construction
\[
\ell_1\leq k-1=\lfloor(N+1)/2\rfloor-1,\text{ and } \ell_2\leq N-k = N-\lfloor (N+1)/2\rfloor.
\]
We  set $\Phi_N(\lambda)=(\alpha,\beta,\lambda_k)$ where $\alpha=(\alpha_1,\ldots,\alpha_{\ell_1})$ and $\beta = (\beta_1,\ldots,\beta_{\ell_2})$. The expression of $\alpha_i$ and $\beta_i$ yields the injectivity of $\Phi_N$. Reciprocally, from $(\alpha,\beta,n)\in\Lambda_N$ we can construct $\lambda\in\widehat{\U}(N)$ with \eqref{eq:lambda_odd}, which yields the surjectivity of $\Phi_N$, hence $\Phi_N$ is indeed a bijection.
\end{proof}

In the sequel, given a pair of partitions $(\alpha,\beta)$ and an integer $n$, for any $N$ such that $(\alpha,\beta,n) \in \Lambda_N,$ we will denote by $\lambda_N(\alpha,\beta,n)=\Phi_N^{-1}(\alpha,\beta,n)$ the highest weight constructed through $\Phi_N^{-1}.$ A crucial result is that $c_2(\lambda)$ has then an explicit expression as a function of $\alpha, \beta, n$ and $N$ that will be convenient for asymptotic analysis. This expression involves the total content $K$ introduced in \eqref{def:content}.

\begin{proposition}[\cite{Lem}]\label{prop:cas_alphabeta}
For any $N \ge 1$ and $(\alpha,\beta,n) \in \Lambda_N,$  we have
\begin{equation}\label{eq:cas_alphabeta_TL}
c_2(\lambda_N(\alpha,\beta,n)) = \vert\alpha\vert+\vert\beta\vert + n^2 + \frac{2}{N}F(\alpha,\beta,n),
\end{equation}
where
\[
F(\alpha,\beta,n) = K(\alpha)+K(\beta)+n(\vert\alpha\vert-\vert\beta\vert).
\]
\end{proposition}

We are now ready to exhibit the relationship between $\Gfr_N(q)$, $\Gfr_1(q)$ and $\Ufr(q)$. It can be roughly stated as follows: the pushforward of $\Gfr_N(q)$ by the map $\Phi_N$ is absolutely continuous with respect to the product measure $\Ufr(q)^{\otimes 2}\otimes\Gfr_1(q)$, with an explicit density. Here is the exact result:

\begin{proposition}\label{prop:transfert}
Let $N\geq 1$ be an integer and $q\in(0,1)$ be a real number. If $\lambda\sim\Gfr_N(q),$ for any bounded measurable $f:\Lambda_N\to\R$,  we have
\begin{equation}\label{eq:transfert}
\E[f\circ\Phi_N(\lambda)] = \frac{\theta(q)}{Z_N(q)\phi(q)^2}\E\left[f(\alpha,\beta,n)q^{\frac{2}{N}F(\alpha,\beta,n)}\mathbf{1}_{\Lambda_N}(\alpha,\beta,n)\right],
\end{equation}
where $\alpha,\beta,n$ are independent, $\alpha,\beta\sim\Ufr(q)$ and $n\sim\Gfr_1(q)$.
\end{proposition}

\begin{proof}
It is straightforward from the definitions of the measures and Proposition \ref{prop:cas_alphabeta}.
\end{proof}

\section{Proofs of the main theorems} \label{sec:largeN}

\subsection{The large $N$ limit of $\Gfr_N(q)$}

The behavior of $\Gfr_N(q)$ (or more precisely of the pushforward of $\Gfr_N(q)$ by the map $\Phi_N$) in the large $N$ limit is given by Theorem \ref{thm:cv_distrib}, that we will prove in the present subsection.

Before that, we need some preliminary control on the Casimir numbers. The main interest of this result is to give a minoration of $c_2$ which involves only the integers $n$ and $N$ and the sizes of the two partitions $\alpha$ and $\beta.$

\begin{lemma}[Domination lemma]\label{lem:domination2}
For any $(\alpha,\beta,n)\in\Lambda_N$, for $N \geq 1,$ 
\begin{equation}\label{eq:domination3}
c_2(\lambda_N(\alpha,\beta,n)) \geq \frac12(\vert\alpha\vert+\vert\beta\vert)+\left(n+\frac{\vert\alpha\vert-\vert\beta\vert}{N}\right)^2=:C_N(\alpha,\beta,n).
\end{equation}
Furthermore, for any $q\in(0,1)$,
\[
\sup_{N \geq 1} \sum_{(\alpha,\beta,n)\in\Lambda_N}q^{C_N(\alpha,\beta,n)} <\infty.
\]

\end{lemma}

\begin{proof}
Let $(\alpha,\beta,n)\in\Lambda_N$. We have
\begin{align*}
c_2(\lambda_N(\alpha,\beta,n)) = & \vert\alpha\vert + \vert\beta\vert + n^2 +\frac{2}{N}(K(\alpha)+K(\beta) + n(\vert\alpha\vert-\vert\beta\vert))\\
= &  \vert\alpha\vert + \vert\beta\vert + \frac{2}{N}(K(\alpha)+K(\beta)) + \left(n+\frac{\vert\alpha\vert-\vert\beta\vert}{N}\right)^2 - \frac{1}{N^2}(\vert\alpha\vert-\vert\beta\vert)^2\\
\geq & \vert\alpha\vert + \vert\beta\vert + \frac{2}{N}(K(\alpha)+K(\beta)) + \left(n+\frac{\vert\alpha\vert-\vert\beta\vert}{N}\right)^2 - \frac{1}{N^2}(\vert\alpha\vert^2+\vert\beta\vert^2)
\end{align*}

Using the Cauchy--Schwarz inequality, we have 
\[
\frac{1}{N^2}\vert\alpha\vert^2 \leq \left(\sum_{k=1}^{\ell(\alpha)}\frac{1}{N^2}\right)\left(\sum_{k=1}^{\ell(\alpha)}\alpha_k^2\right)=\frac{\ell(\alpha)}{N^2}\sum_{k=1}^{\ell(\alpha)}\alpha_k^2\leq\frac{1}{2N}\sum_k\alpha_k^2,
\]
where, for the last inequality, we have used that $\ell(\alpha)\leq N/2.$
It yields
\begin{align*}
\frac{2}{N}K(\alpha) - \frac{1}{N^2}\vert\alpha\vert^2 \geq & \frac1N\left(\sum_{i=1}^{\ell(\alpha)} \alpha_i(\alpha_i+1-2i) - \frac12 \sum_{i=1}^{\ell(\alpha)} \alpha_i^2\right)\\
\geq & \frac1N \sum_{i=1}^{\ell(\alpha)} \alpha_i(1-2i)\\
\geq & - \frac{\vert\alpha\vert\ell(\alpha)}{N}.
\end{align*}
Using the same bound for $\beta$ leads to the following inequality:
\[
c_2(\lambda_N(\alpha,\beta,n) \geq \vert\alpha\vert\left(1-\frac{\ell(\alpha)}{N}\right)+\vert\beta\vert\left(1-\frac{\ell(\beta)}{N}\right)+\left(n+\frac{\vert\alpha\vert-\vert\beta\vert}{N}\right)^2.
\]
However, by assumption, we have $\ell(\alpha)\leq A_N,\ell(\beta)\leq B_N$, and for $N  \geq 1,$ we have
\[
\frac{\max(A_N,B_N)}{N} \leq \frac{1}{2}.
\]
We thus obtain \eqref{eq:domination3}. Now, let us bound
\[
\sum_{(\alpha,\beta,n)\in\Lambda_N}q^{C_N(\alpha,\beta,n)}.
\]
We have
\begin{align*}
\sum_{(\alpha,\beta,n)\in\Lambda_N}q^{C_N(\alpha,\beta,n)} & = \sum_{\substack{\ell(\alpha)\leq A_N\\ \ell(\beta)\leq B_N}} q^{\frac12(\vert\alpha\vert+\vert\beta\vert)}\left(\sum_{n\in\Z} q^{\left(n+\frac{\vert\alpha\vert-\vert\beta\vert}{N}\right)^2}\right).
\end{align*}
The sum between the brackets can be bounded independently from $N,\alpha,\beta$ by
\begin{equation}
\label{eq:quasiinvariance}
C:=\sup_{x\in[0,1]}\sum_{n\in\Z} q^{(n+x)^2} <\infty,
\end{equation}
and therefore 
\[ \sup_{N \ge 1  } \sum_{(\alpha,\beta,n)\in\Lambda_N}q^{C_N(\alpha,\beta,n)}  \leq \frac{C}{\phi(q^{1/2})^{2}}.\]
\end{proof}

\begin{remark}
The quasi-periodicity \eqref{eq:quasiinvariance} of the theta function is critical in the proof of the domination lemma. There might be an interpretation in terms of quasimodular forms.
\end{remark}

We are now in position to prove the asymptotic decorrelation. As a matter of fact, we will simultaneously prove Theorem \ref{thm:lim_pf} and \ref{thm:cv_distrib}.

\begin{proof}[Proof of Theorem \ref{thm:cv_distrib} and Theorem \ref{thm:lim_pf}]
Let $f:\Pfr^2\times\Z\to\R$ be measurable and bounded. We have to study the convergence of the following sum:
\[
\sum_{(\alpha, \beta,n) \in \Pfr^2\times\Z} G_N(\alpha, \beta,n),
\]
with
\[
 G_N(\alpha, \beta,n):= f(\alpha, \beta,n) q^{c_2(\lambda_N(\alpha, \beta,n))} \mathbf 1_{\Lambda_N}(\alpha, \beta,n).
\]
It is clear that, for any $(\alpha, \beta,n) \in \Pfr^2\times\Z,$ we have the pointwise convergence
\[
G_N(\alpha, \beta,n) \xrightarrow[N \rightarrow \infty]{} f(\alpha, \beta,n) q^{|\alpha|}q^{|\beta|}q^{n^2}:= g(\alpha, \beta,n).
\]

As we are working with measures of infinite mass, the dominated convergence criterion that we will apply is not standard. For the sake of completeness, we state and prove it in the Appendix.
The proof therefore boils down to show that the family of functions $(G_N-g)$ satisfies \eqref{eq:Bogachev}, thanks to Lemma \ref{lem:cv_L1}. We present in details the estimates for $(G_N).$ The corresponding estimates for $g$ are straightforward.

For any positive integer $M$, set
\[
X_{M} = \{(\alpha,\beta,n)\in\Pfr^2\times\Z:\ \vert\alpha\vert\leq M, \vert\beta\vert\leq M, \vert n\vert \leq M\}.
\]
It is always a finite set so it has finite counting measure. We first have :
\[
\sum_{(\alpha, \beta,n) \in X_{M}^c} G_N(\alpha, \beta,n) \leq \|f\|_{\infty}
\sum_{(\alpha, \beta,n)\in X_{M}^c\cap\Lambda_N} q^{c_2(\lambda_N(\alpha, \beta,n))}, 
\]
we now use Lemma \ref{lem:domination2} to get 
\[
\sum_{(\alpha, \beta,n) \in X_{M}^c} G_N(\alpha, \beta,n) \leq \|f\|_{\infty}
\sum_{(\alpha, \beta,n) \in X_{M}^c\cap\Lambda_N} q^{C_N(\alpha, \beta,n)}. 
\]
Let us show that the sum in the right-hand side can be arbitrarily small for $M$ large enough. By the same arguments as in the proof of Lemma \ref{lem:domination2},
\begin{align*}
\sum_{(\alpha,\beta,n)\in\Lambda_N}&\mathbf{1}_{X_M}(\alpha,\beta,n) q^{C_N(\alpha,\beta,n)}\\
= &\sum_{\substack{\ell(\alpha)\leq A_N\\ \ell(\beta)\leq B_N}} \mathbf{1}_{\{\vert\alpha\vert > M,\vert\beta\vert > M\}} q^{\frac12(\vert\alpha\vert+\vert\beta\vert)}\left(\sum_{n\in\Z}\mathbf{1}_{\{\vert n\vert > M\}}q^{\left(n+\frac{\vert\alpha\vert-\vert\beta\vert}{N}\right)^2}\right).
\end{align*}
As all terms are nonnegative, it is clear that the right-hand side is bounded by
\begin{align*}
\sum_{\vert\alpha\vert,\vert\beta\vert  > M} q^{\frac12(\vert\alpha\vert+\vert\beta\vert)}\left(\sum_{n\in\Z}q^{\left(n+\frac{\vert\alpha\vert-\vert\beta\vert}{N}\right)^2}\right) &\leq C\sum_{\vert\alpha\vert,\vert\beta\vert > M} q^{\frac12(\vert\alpha\vert+\vert\beta\vert)} \\
& \leq C \phi(q^{\frac{1}{2}})^2 \Pbb(|\alpha| >M)^2,
\end{align*}
where $\alpha \sim \Ufr(q)$ and $C$ is defined in \eqref{eq:quasiinvariance}. Using Lemma \ref{lem:remainder_pn2}, we get that, for $M$ large enough, it is smaller than
\[
\frac{\varepsilon}{C\|f\|_\infty},
\]
and it follows that
\[
\sup_{N\geq 1}\sum_{(\alpha, \beta,n) \in X_{M}^c} G_N(\alpha, \beta,n) \leq \varepsilon,
\]
as required.

Note that we have actually proved the convergence for the non-normalized measure rather than $\Gfr_N(q)$; if we apply the result to $f(\alpha,\beta,n)=1$ we obtain the convergence of the partition function, which is Theorem \ref{thm:lim_pf}, and the proof of Theorem \ref{thm:cv_distrib} follows directly.
\end{proof}

\subsection{Asymptotic expansion of the partition function}

Our goal is now to prove Theorem~\ref{thm:main}. Before going to its proof, we will again need some preliminary lemmas. For any $k,\ell\geq 0$ and any $q\in(0,1)$, we set
\[
A_{k,\ell} = \E\left[K(\alpha)^k\vert\alpha\vert^\ell\right] \text{ and } B_{k,\ell} = \E\left[n^{k+\ell}\right],
\]
for $\alpha\sim\Ufr(q)$ and $n\sim\Gfr_1(q)$.

\begin{lemma}\label{lem:finite_pow}
Let $q\in(0,1)$, and $\alpha,\beta\sim\Ufr(q)$ and $n\sim\Gfr_1(q)$ be independent random variables. 
For any $k\geq 0,$ we have
\begin{equation}\label{eq:moments_even}
\E\left[F(\alpha,\beta,n)^{2k}\right] = \sum\binom{2k}{2k_1\ 2k_2\ k_3\ k_4}(-1)^{k_4}A_{2k_1,k_3}A_{2k_2,k_4}B_{k_3,k_4},
\end{equation}
where the sum is over $k_1,k_2,k_3,k_4\geq 0$ such that $2k_1+2k_2+k_3+k_4 = 2k$,
and
\begin{equation}\label{eq:moments_odd}
\E\left[F(\alpha,\beta,n)^{2k+1}\right]=0.
\end{equation}
\end{lemma}

\begin{proof}
The proof will be organized as follows. First, we expand $F(\alpha,\beta,n)^k$ for any $k$ using the multinomial formula, then we prove that the expectation for odd $k$ is zero, which will simplify the general formula.

\noindent{\it Step 1: expansion.} From the definition of $F(\alpha,\beta,n)$ and the multinomial formula, we have
\[
\E[F(\alpha,\beta,n)^k] = \sum \binom{k}{k_1\ k_2\ k_3\ k_4} \E\left[ K(\alpha)^{k_1}K(\beta)^{k_2} n^{k_3+k_4}\vert\alpha\vert^{k_3}\vert\beta\vert^{k_4}(-1)^{k_4} \right].
\]
where the sum is over $k_1,k_2,k_3,k_4\geq 0$ such that $k_1+k_2+k_3+k_4=k$. By independence and linearity, it rewrites
\begin{equation}\label{eq:first_exp_fk}
\E[F(\alpha,\beta,n)^k] = \sum_{\substack{k_1,k_2,k_3,k_4\geq 0\\ k_1+k_2+k_3+k_4 = k}} \binom{k}{k_1\ k_2\ k_3\ k_4} (-1)^{k_4} A_{k_1,k_3} A_{k_2,k_4} B_{k_3,k_4}.
\end{equation}

\noindent{\it Step 2: cancellations.} If $\lambda$ is a partition, we denote by $\lambda'$ its conjugate, in the sense that $\lambda'_i = \# \{j, \lambda_j \ge i \}.$ As $\lambda \mapsto \lambda'$ is a bijection on $\Pfr$ and $\vert \lambda\vert = \vert \lambda'\vert,$ we have that, if $\alpha \sim \Ufr(q)$ then $\alpha' \sim \Ufr(q).$ Moreover, 
$K(\lambda') = - K(\lambda)$ so that 
\[ \E\left[K(\alpha)^{k_1}\vert\alpha\vert^{k_3} \mathbf{1}_{\{K(\alpha)>0\}}\right] =(-1)^{k_1}\E\left[K(\alpha)^{k_1}\vert\alpha\vert^{k_3}\mathbf{1}_{\{K(\alpha)<0\}}\right] \]
and 
\[
\E\left[K(\alpha)^{k_1}\vert\alpha\vert^{k_3}\right] = (1+(-1)^{k_1})\E\left[K(\alpha)^{k_1}\vert\alpha\vert^{k_3}\mathbf{1}_{\{K(\alpha)>0\}}\right],
\]
which cancels when $k_1$ is odd.
By a similar argument for $k_2,$ we see that we can restrict the sum in \eqref{eq:first_exp_fk} to
\[
\{k_1,k_2,k_3,k_4\geq 0\vert\ k_1+k_2+k_3+k_4=k, k_1 \text{ and } k_2 \text{ are even}\}.
\]
Now, observe that for any $k_1,k_2,k_3,k_4$, the terms
\[
(-1)^{k_4}A_{k_1,k_3}A_{k_2,k_4}B_{k_3,k_4}
\]
and
\[
(-1)^{k_3}A_{k_2,k_4}A_{k_,k_3}B_{k_4,k_3}
\]
cancel each other if $k_3$ and $k_4$ do not have the same parity. Therefore $k$ has to be even and we obtain both \eqref{eq:moments_even} and \eqref{eq:moments_odd}.
\end{proof}

\begin{lemma}\label{lem:concentration_pow}
Let $q\in(0,1)$, and $\alpha,\beta\sim\Ufr(q)$ and $n\sim\Gfr_1(q)$ be independent random variables. 
For any $k\geq 1,$ there exists a constant $C_{k,q}$ such that 
\[
\E\left[F(\alpha,\beta,n)^k\mathbf{1}_{{\Lambda}_N^c}(\alpha,\beta,n)\right]\leq C_{k,q} e^{\frac{N}{2}\log q}.
\]
Consequently, for any $p\in\N$ and $k\geq 1,$ there is a constant $C(k,p,q)>0$ such that, for any $N\ge 1,$ 
\begin{equation}\label{eq:concentration_pow}
\E\left[F(\alpha,\beta,n)^k\mathbf{1}_{{\Lambda}_N^c}(\alpha,\beta,n)\right]\leq C(k,p,q) N^{-p}.
\end{equation}
\end{lemma}

\begin{proof}
As $\alpha$ and $\beta$ are playing a similar role, we have
\[
\E\left[F(\alpha,\beta,n)^k\mathbf{1}_{\Lambda_N^c}(\alpha,\beta,n)\right]\leq 2  \E\left[F(\alpha,\beta,n)^k\mathbf{1}_{\{\ell(\alpha)>\lfloor(N-1)/2\rfloor\}}\right].
\]

By Cauchy-Schwarz inequality, we have 
\[ 
 \E\left[F(\alpha,\beta,n)^k\mathbf{1}_{\{\ell(\alpha)>\lfloor(N-1)/2\rfloor\}}\right]^2\le \Pbb(\ell(\alpha)>\lfloor (N-1)/2\rfloor)
 \E\left[F(\alpha,\beta,n)^{2k} \right],
\]
which is $O(N^{-p})$ with a bound that only depends on $k,p$ and $q$, according to Proposition \ref{prop:tail_pt} and Lemma \ref{lem:finite_pow}.
\end{proof}

We are now in position to prove Theorem \ref{thm:main}, whose precise statement is the following:
\begin{theorem}\label{thm:main3}
Let $t>0$ be a positive real number, and set $q_t=e^{-\frac{t}{2}}$. For any $p\geq 0$ and any integer $N$ large enough, the partition function $Z_N(q_t)$ admits the asymptotic expansion
\begin{equation}\label{eq:zn_expansion}
Z_N(q_t) = \sum_{k=0}^p\frac{a_{2k}(t)}{N^{2k}} + O(N^{-2p-2}),
\end{equation}
where $a_{2k}(t)$ is given by
\[
a_{2k}(t) = \frac{t^{2k}\theta(q_t)}{(2k)!\phi(q_t)^2}\E[F(\alpha,\beta,n)^{2k}],
\]
with $\alpha, \beta$ and $n$  independent, $\alpha, \beta \sim \Ufr(q_t)$ and $n \sim \Gfr_1(q_t).$
\end{theorem}

\begin{remark}
This asymptotic expansion is also a topological expansion in the sense of Section \ref{subsec:topo}. We do not address this aspect in the following proof but we will pursue in this direction in Section \ref{sec:expansions}. Indeed, the proofs that there is a convergent asymptotic expansion and that the coefficients of the expansion have a topological meaning are completely independent.
\end{remark}

\begin{proof}
In the whole proof, $\alpha, \beta$ and $n$ are independent, with $\alpha, \beta \sim \Ufr(q_t)$ and $n \sim \Gfr_1(q_t).$
Let us start by writing the partition function with the help of Proposition \ref{prop:transfert}:
\[
Z_N(q_t) = \frac{\theta(q_t)}{\phi(q_t)^2}\E\left[e^{-\frac{t}{N}F(\alpha,\beta,n)}\mathbf{1}_{\Lambda_N}(\alpha,\beta,n)\right].
\]

We perform a Taylor expansion of the exponential up to order $2p+1$:
\[
e^{-\frac{t}{N}F(\alpha,\beta,n)}= \sum_{k=0}^{2p+1}\frac{(-t)^k}{k!N^k}F(\alpha,\beta,n)^k + \frac{F(\alpha,\beta,n)^{2p+2}}{(2p+1)!N^{2p+2}}I_{N,t}(\alpha,\beta,n),
\]
where
\[
I_{N,t}(\alpha,\beta,n) = \int_0^t (t-s)^{2p+1}e^{-\frac{s}{N}F(\alpha,\beta,n)}ds.
\]
By linearity, it follows that
\begin{align*}
Z_N(q_t) = & \frac{\theta(q_t)}{\phi(q_t)^2}\left(\sum_{k=0}^{2p+1} \frac{(-t)^k}{k!N^k}\E[F(\alpha,\beta,n)^k] - \sum_{k=0}^{2p+1}\frac{(-t)^k}{k!N^k}\E[F(\alpha,\beta,n)^k\mathbf{1}_{\Lambda_N^c}(\alpha,\beta,n)]\right)\\
& + \frac{\theta(q_t)}{\phi(q_t)^2}\frac{1}{(2p+1)!N^{2p+2}}\E[F(\alpha,\beta,n)^{2p+2}I_{N,t}(\alpha,\beta,n)\mathbf{1}_{\Lambda_N}(\alpha,\beta,n)].
\end{align*}
As a consequence of Lemma \ref{lem:concentration_pow}, we have
\begin{equation}\label{eq:remainder1}
\sum_{k=0}^{2p+1}\frac{(-t)^k}{k!N^k}\E[F(\alpha,\beta,n)^k\mathbf{1}_{\Lambda_N^c}(\alpha,\beta,n)] = O(N^{-2p-2}).
\end{equation}
Now we conclude the control of the remainder by claiming that
\begin{equation}\label{eq:remainder2}
\sup_{n}\E[F(\alpha,\beta,n)^{2p+2}I_{N,t}(\alpha,\beta,n)\mathbf{1}_{\Lambda_N}(\alpha,\beta,n)] <\infty.
\end{equation}
Indeed, for any $s \le t,$
\begin{align*}
\E[F(\alpha,\beta,n)^{2p+2}& e^{-\frac{s}{N}F(\alpha,\beta,n)}\mathbf{1}_{\Lambda_N}(\alpha,\beta,n)] \\ 
&\leq \E[F(\alpha,\beta,n)^{2p+2} (1+ e^{-\frac{t}{N}F(\alpha,\beta,n)})\mathbf{1}_{\Lambda_N}(\alpha,\beta,n)] \\
 \leq \E[F(\alpha,\beta,n)^{2p+2}] &+  \E[F(\alpha,\beta,n)^{6p+6}]^{1/3} \E[e^{-\frac{3t}{2N}F(\alpha,\beta,n)}\mathbf{1}_{\Lambda_N}(\alpha,\beta,n)]^{2/3},
\end{align*}
where we have used H\"older inequality with conjugate coefficients $3$ and $3/2.$ Now, by the same argument as in the proof of Lemma \ref{lem:domination2}, the last term can be bounded as follows:
\[
\sum_{(\alpha, \beta,n) \in \Lambda_N} q_t^{\alpha + \beta +n^2 + \frac{3}{N} F(\alpha,\beta,n)} 
\leq \sum_{\substack{\ell(\alpha)\leq A_N\\ \ell(\beta)\leq B_N}} q_t^{\frac14(\vert\alpha\vert+\vert\beta\vert)}\left(\sum_{n\in\Z} q_t^{\left(n+\frac{3\vert\alpha\vert-\vert\beta\vert}{2N}\right)^2}\right),
\]
which is finite as explained in the proof of Lemma \ref{lem:domination2}. We therefore get that 
\[ Z_N(q_t) = \frac{\theta(q_t)}{\phi(q_t)^2} \sum_{k=0}^{2p+1} \frac{(-t)^k}{k!N^k}\E[F(\alpha,\beta,n)^k]  + O(N^{-2p-2})
\]
and we conclude by using \eqref{eq:moments_odd}.
\end{proof}

\section{Random partitions and enumeration of ramified coverings}
\label{sec:enumeration}

In this section, we  deal with the enumeration of ramified coverings of compact Riemann surfaces, with a point of view inspired by Okounkov's works \cite{Oko2,BO,EO}. It is based on the idea that one can turn generating series of the generalized Hurwitz numbers into probability measures on partitions, making an effective link between models of random partitions and enumerative geometry. We start by recalling the relationship between partitions and enumeration of ramified coverings, then we introduce and study new measures on partitions, that we call respectively the \emph{Frobenius} and the \emph{Plancherel--Hurwitz} measures. The main point is that these measure are related to the $q$-uniform measure and can benefit from the results we obtained in Section \ref{sec:phi_theta}. We finish this section by a topological interpretation of the expectations involved in the coefficients of the asymptotic expansion \eqref{eq:zn_expansion}.

\subsection{Enumeration of ramified coverings}

If $X,Y$ are two topological spaces, a continuous surjective map $p:Y\to X$ is a \emph{covering} of $X$ if for any $x\in X$, there is an open neighbourhood $U$ of $x$ such that preimage $p^{-1}(U)$ is the disjoint union of open sets $V_1,\ldots,V_n$, and $p\vert_{V_i}:V_i\to U$ is a homeomorphism for all $1\leq i\leq n$. If $X$ is a smooth path-connected surface and $R\subset X$ is a finite subset of $X$, a \emph{ramified covering} of $X$ with \emph{ramification locus} $R$ is the datum of a topological space $Y$ and a continuous mapping $p:Y\to X$ such that:
\begin{enumerate}
\item the mapping $p\vert_{p^{-1}(X\setminus R)}:p^{-1}(X\setminus R)\to X\setminus R$ is a covering;
\item for all $x\in X$ and $y\in p^{-1}(\{x\})$, there exists an open neighbourhood $V$ of $y$ and an integer $n\geq 1$ such that the mapping $p\vert_V$ is (topologically) conjugated to the mapping
\[
\left\lbrace\begin{array}{ccc}
\C & \longrightarrow & \C\\
z & \longmapsto & z^k
\end{array}\right..
\]
\end{enumerate}
In this case, $k$ is called the \emph{ramification index}, or the \emph{ramification order}, of $X$ at $x$. Equivalently, the preimage of each branch point by $p$ is given by $n-(k-1)$ points. We also denote by $B$, and call \emph{branch locus} of $X$, the set $p(R)$. Its elements are called \emph{branch points}. In what follows, we will only consider \emph{generic} ramifications, i.e. the ramification index of $X$ at each branch point will be 2. 

A powerful result is that if the base space is a Riemann surface, then the ramified covering can be equipped with a complex structure such that $p$ is holomorphic\footnote{This fact was one of the initial motivations from Eskin and Okounkov, because it enables to link the enumeration of ramified coverings to the computation of Masur--Veech volumes, see \cite{Agg} and references therein for more recent developments.}. In particular it is orientable and it has a well-defined genus. The set of $n$-sheeted coverings of a closed orientable surface $X$ with ramification locus $R$ is in bijection with the representation space of the fundamental group $\pi_1(X\setminus R)$ into the symmetric group $\mathfrak{S}_n$ \cite{EO}:
\begin{equation}\label{eq:equiv_cov}
\{n\text{-sheeted ramified coverings of } X\setminus R\} \simeq \Hom(\pi_1(X\setminus R),\mathfrak{S}_n)/\mathfrak{S}_n.
\end{equation}

The standard way to count ramified coverings is therefore by counting weighted isomorphism classes of the right-hand side of \eqref{eq:equiv_cov}, the weight being given by $\frac{1}{\vert\mathfrak{S}_n\vert} = \frac{1}{n!}$. Assuming that we only consider generic ramifications and that $X$ has genus $g$, the corresponding number of ramified coverings is given by the \emph{Hurwitz numbers} with base $X$:
\begin{equation}\label{eq:HurwitzNumbers}
H_g(n,k) = \sum_{\alpha\vdash n}\left(\frac{n!}{\chi_\alpha(1)}\right)^{2g-2}K(\alpha)^k.
\end{equation}
A detailed derivation of these numbers can be found in \cite[Appendix A]{LZ}. By the Riemann--Hurwitz formula, if $Y$ is a $n$-sheeted ramified covering over a surface $X$ of genus $g$ with ramification locus $R=\{x_1,\ldots,x_k\}$, then the genus $g'$ of $Y$ is given by
\begin{equation}\label{eq:RiemannHurwitz}
g' = \frac{k}{2} + 1 + n(g-1).
\end{equation}
Here, we are only interested in the particular case of $g=1$ (when the base is a torus). We obtain a nice simple relationship between the number of ramification points and the genus of the ramified covering: $k=2g'-2$. Besides, in order to ensure that the ramified covering corresponds to an orientable surface, we see that $k$ has to be even. In order to stress this fact, we will enumerate the Hurwitz numbers as $H_g(n,2g'-2)$ with $g'\geq 1$ in the sequel. Their generating function is the formal power series
\begin{equation}\label{eq:Gen_Fun_Hurwitz}
\F_{g,g'}(q) = \sum_{n\geq 1} H_g(n,2g'-2)q^n.
\end{equation}
Let us make a few comments on how to compute these generating functions when the base space is a torus.
\begin{itemize}
\item The case $g'=1$ is trivial:
\[
\F_{1,1}(q)=\sum_{n\in\N}p(n)q^n=\phi(q)^{-1}.
\]
\item For $g'\geq 2$, Kaneko and Zagier \cite{KZ} proved that $\F_{g,g'}$ is a quasimodular form of weight $6k$, inspired by works of Douglas \cite{Dou} and Dijkgraaf \cite{Dij} on mirror symmetry and conformal field theory. See \cite[Corollary A.2.17]{LZ} for a precise statement and a few explicit computations.
\end{itemize}

\subsection{New measures on partitions}

Using formulas \eqref{eq:Gen_Fun_Hurwitz} and \eqref{eq:HurwitzNumbers}, we can define a measure on $\Pfr$ with parameters $g\in\N$, $g'\in\N^*$ and $q\in(0,1)$ by
\[
\Pbb(\alpha) \propto \left(\frac{\chi_\alpha(1)}{\vert\alpha\vert!}\right)^{2-2g}K(\alpha)^{2g'-2}q^{\vert\alpha\vert}.
\]
We call it the \emph{Frobenius measure} of genera $g$ and $g'$ and parameter $q$ and denote it by $\Ffr_{g,g'}(q)$. Its partition function is simply the generating function $\F_{g,g'}(q)$. We made the choice of only considering even powers of the total content of partitions in order to avoid signed measures, but a similar construction is possible with any power of $K(\alpha)$. The Frobenius measure is related to another measure $\PH_{g,g'}(n)$ on $\Pfr_n$, that generalizes the Plancherel--Hurwitz measure introduced by Chapuy, Louf and Walsh \cite{CLW}, corresponding to the case $g=0$. It is defined by
\[
\Pbb(\alpha)\propto \left(\frac{\chi_\alpha(1)}{\vert\alpha\vert!}\right)^{2-2g}K(\alpha)^{2g'-2},
\]
and its partition function is simply $H_g(n,2g'-2).$ We have the following relationship between $\Ffr_{g,g'}(q)$ and $\PH_{g,g'}(n)$, reminiscent of the ones for the Poissonized Plancherel measure $(g=0,g'=1)$ and the $q$-uniform measure $(g=g'=1)$.

\begin{proposition}\label{prop:poissonisation}
For any $g\geq 0$, $g'\geq 1$ and any $q\in(0,1)$, if $\alpha\sim\Ffr_{g,g'}(q)$, then the distribution of $\vert\alpha\vert$ is given by
\begin{equation}
\Pbb(\vert\alpha\vert = n) \propto H_g(n,2g'-2) q^n.
\end{equation}
Besides, conditionally to $\vert\alpha\vert = n$, the distribution of $\alpha$ on $\Pfr_n$ is given by $\PH_{g,g'}(n)$.
\end{proposition}

\begin{proof}
It is the same proof as in Proposition \ref{prop:phi_condition}: the first assertion comes from the averaging over all partitions with the same size, and the second one from the definition of the conditional probability.
\end{proof}

Here is an elementary result, whose proof is identical to Proposition \ref{prop:is_it_wick}.
\begin{proposition}\label{prop:moments_frob}
For any $g\geq 0$, $g'\geq 1$ and $q\in(0,1)$, if $\alpha\sim\Ffr_{g,g'}(q)$ is a random partition and $\ell\geq 0$ is an integer, then
\begin{equation}\label{eq:moment_frobenius1}
\E\left[\vert\alpha\vert^\ell\right] = \frac{D^\ell\F_{g,g'}(q)}{\F_{g,g'}(q)}.
\end{equation}
\end{proposition}

Let us comment \eqref{eq:moment_frobenius1} in a few simple cases.

\begin{itemize}
\item If $\ell=1$, we find the average size of a $\Ffr_{g,g'}(q)$-distributed random partition:
\[
\E\left[\vert\alpha\vert\right] = \frac{q\frac{d}{dq}\F_{g,g'}(q)}{\F_{g,g'}(q)} = q\frac{d}{dq}\log\F_{g,g'}(q).
\]
\item If $\ell=g=g'=1$, we retrieve a result by Okounkov \cite{Oko}:
\begin{equation}
\phi(q)\sum_{\alpha\in\Pfr} \vert\alpha\vert q^{\vert\alpha\vert} =  q\frac{d}{dq}\log\phi(q) = \frac{1}{24} + G_2(q),
\end{equation}
where $G_2(q)=\frac12\zeta(-1)+\sum_n q^n\sum_{d\vert n}d$ is an Eisenstein series.
\end{itemize}

Although it might be of independent interest to study in detail the Plan\-cherel--Hurwitz and the Frobenius measures, we will not go further in this direction. We will simply link the expectations from \eqref{eq:moment_frobenius1} to the expectations $A_{k,\ell}$ that appear in Lemma \ref{lem:finite_pow}: recall that if $\alpha$ is a $q$-uniform random partition,
\[
A_{k,\ell} = \E\left[K(\alpha)^k\vert\alpha\vert^\ell\right].
\]
It follows that for any integer $g\geq 1$,
\[
A_{2g-2,\ell} = \phi(q)\sum_{\alpha\in\Pfr} K(\alpha)^{2g-2}\vert\alpha\vert^\ell q^{\vert\alpha\vert}=\phi(q)\F_{1,g}(q)\E\left[\vert\beta\vert^\ell\right],
\]
where $\beta\sim\Ffr_{1,g}(q)$. Finally,
\begin{equation}\label{eq:A_kl}
A_{2g-2,\ell} = \phi(q)D^{\ell}\F_{1,g}(q).
\end{equation}

Now we shall state the main result of this subsection.

\begin{theorem}\label{thm:moments_F}
For any any $k\geq 0$ and $q\in(0,1)$, if $(\alpha,\beta,n)$ is a triple of independent random variables with $\alpha,\beta\sim\Ufr(q)$ and $n\sim\Gfr_1(q)$, we have
\begin{align}
\begin{split}
\E & \left[F(\alpha,\beta,n)^{2k}\right] = \\
&\sum\frac{(-1)^{k_4}(2k)!\phi(q)^2}{(2k_1)!(2k_2)!k_3!k_4!\theta(q)}D^{k_3}\F_{1,k_1+1}(q)D^{k_4}\F_{1,k_2+1}(q)D^{\frac{k_3+k_4}{2}}\theta(q),
\end{split}\label{eq:general_moment}
\end{align}
where the sum is over $k_1,k_2,k_3,k_4\geq 0$ such that $2k_1+2k_2+k_3+k_4=2k$.
\end{theorem}

\begin{proof}
This is a direct consequence of \eqref{eq:moments_even}, \eqref{eq:moment_frobenius1}, \eqref{eq:theta_moment_even} and \eqref{eq:A_kl}.
\end{proof}

\section{Topological expansions}
\label{sec:expansions}

In this section, we combine the asymptotic expansions from Section \ref{sec:largeN} and the results from Section \ref{sec:enumeration} in order to get the topological expansion of the Yang--Mills partition function on a torus. We will first focus on a special case, the so-called \emph{chiral setting}, introduced by Gross and Taylor in \cite{GT,GT2}. It consists in replacing, in \eqref{eq:Z-YM}, the set of summation $\widehat{\U}(N)$ by $\{\alpha\in\Pfr,\ell(\alpha)\leq N\}$, and apply the dimension and Casimir formulas to the highest weight obtained from $\alpha$ by adding the right amount of zeros: $\tilde{\alpha}=(\alpha_1,\ldots,\alpha_{\ell(\alpha)},0,\ldots,0)\in\widehat{\U}(N)$. More precisely,
\[
Z_N^{\chir}(g,t) = \phi(q)^{-1}\E\left[d_{\tilde{\alpha}}^{2-2g} q_t^{\frac{2}{N}K({\alpha})} \mathbf{1}_{\{\ell(\alpha)\leq N\}}\right],
\]
where $d_{\tilde{\alpha}}$ is the dimension of the associated representation. As we treat here the case  $g=1,$ we do not need to provide an exact formula for this dimension.

Let us now prove the analog of Theorem \ref{thm:main} in the chiral setting.

\begin{proof}[Proof of Theorem \ref{thm:chiral_pf}]
Using the exact same estimates as in the proof of Theorem \ref{thm:main3}, we have
\[
Z_N^\chir(1,t) = \sum_{k=0}^{2p+1} \frac{(-t)^k}{k!N^k}\E[K(\alpha)^k] + O(N^{-2p-2}),
\]
where $\alpha \sim \Ufr(q_t).$

The odd coefficients vanish because of the identity $K(\alpha)^{2k+1}=-K(\alpha')^{2k+1}$, and the even coefficients can be computed using Proposition \ref{prop:moments_frob}. The result follows.
\end{proof}

Theorem \ref{thm:chiral_pf} has also been proved independently by Novak \cite{Nov2024}, who recovered the power series expansion from Gross and Taylor and proved that it is absolutely convergent. The fact that it is a topological expansion follows from the Riemann--Hurwitz formula \eqref{eq:RiemannHurwitz}. We claim that one can go even further, and recover a beautiful formula, due to Gross and Taylor, written as an integral over the space of ramified coverings. In order to do so, we shall borrow\footnote{We actually rather adapt them, in the sense that L\'evy defined the objects when the base space is the unit disk, but everything remains well defined if we replace it by the torus.} some constructions from \cite{Lev}. Let $X$ be a compact complex torus, i.e. a compact connected Riemann surface of genus 1. Denote by $\mathcal{R}$ (resp. $\mathcal{R}_n$) the set of ramified coverings (resp. $n$-sheeted ramified coverings) $Y\to X$ such that:
\begin{enumerate}
\item $Y$ is also a compact\footnote{Not necessarily connected!} Riemann surface,
\item All ramifications are simple.
\end{enumerate}
If $R=\{x_1,\ldots,x_{2k}\}$ is a finite subset of $X$, denote by $\mathcal{R}_{n,R}$ the set of $n$-sheeted ramified coverings $Y\to X$ with ramification locus $R$. In this case, we say that $n$ is the degree of the ramified covering $Y,$ denoted  by $\mathrm{deg}(Y)$ and the Euler characteristics of $Y$ is denoted  by $\chi(Y).$

\begin{definition}
Let $t>0$ be a fixed positive real number. 
\begin{enumerate}
\item For any $R=\{x_1,\ldots,x_{2k}\}\subset X$, we define the measure $\rho_{n,R}$ on $\mathcal{R}_{n,R}$ by
\[
\rho_{n,R} = \sum_{Y\in\mathcal{R}_n}\frac{n!}{\vert\mathrm{Aut}(Y)\vert}\delta_Y.
\]
\item Assume that $X$ is endowed with a volume form $\vol$ such that $\vol(X)=t$, and that $\Xi_t$ is the distribution of a Poisson point process on $X$ with intensity $\vol$, conditioned to only take even values. We define the measure $\rho_{n,t}$ on $\mathcal{R}_n$ by
\[
\rho_{n,t}(dY)=\int_{\mathrm{Conf}(X)}\rho_{n,R}(dY)\Xi_t(dR),
\]
where $\mathrm{Conf}(X)$ is the set of configurations, i.e. locally finite subsets of $X$.
\item We define the measure $\rho_t$ on $\mathcal{R}$ by
\[
\rho_t(dY) = \sum_{n\geq 0}\mathbf{1}_{\mathcal{R}_n}(Y)\rho_{n,t}(dY).
\]
\end{enumerate}
\end{definition}

Roughly speaking, $\rho_t$ is the uniform counting measure of isomorphism classes of ramified coverings $Y\to X$ with simple ramifications, and with a random ramification locus $R$ whose distribution is given by $\Xi_t$.

We can now rewrite the chiral partition function as an integral over the space of ramified coverings of the torus, recovering an analog\footnote{In their expression, Gross and Taylor actually integrate over a larger set of coverings, which are allowed to have what they call \emph{twists and loops} in addition to the branch points. They also have an additional factor $(-1)^i$ where $i$ denotes the number of branch points of the covering. However, in our case the number of branch points is always even, which justifies the absence of the factor $(-1)^i$.} of the equation (2.27) from \cite{GT2}.

\begin{corollary}\label{cor:integral_covering}
The chiral partition function on a compact torus of area $t>0$ admits the following expression:
\begin{equation}
Z_N^\chir(1,t) = \int_{\mathcal{R}}q_t^{\mathrm{deg}(Y)}N^{\chi(Y)}\rho_t(dY).
\end{equation}
\end{corollary}
\begin{proof}
Let us start by taking the limit $p\to\infty$ in Theorem \ref{thm:chiral_pf}. By Fubini's theorem (all terms are nonnegative),
\begin{align*}
Z_N^\chir(1,t) = & \sum_{n\geq 0}q_t^n\sum_{k\geq 0}\frac{t^{2k}}{(2k)!N^{2k}}H_1(n,2k).
\end{align*}
By the Riemann--Hurwitz formula \eqref{eq:RiemannHurwitz}, a ramified covering $Y$ of the torus has genus $g$ if and only if its Euler characteristic equals minus the number of branch points: $2k=2g-2=-\chi(Y)$. Furthermore, from the definition of the Hurwitz numbers, we have for all $k\geq 0$
\[
H_1(n,2k) = \sum_{\substack{Y\in\mathcal{R}_n\\ \chi(Y)=-2k}}\frac{n!}{\vert\mathrm{Aut}(Y)\vert},
\]
so that
\[
Z_N^\chir(1,t) = \sum_{n\geq 0}q_t^n\sum_{Y\in\mathcal{R}_n}N^{\chi(Y)}\int_{\mathrm{Conf}(X)} \rho_{n,R}(dY)\Xi_t(dR).
\]
The result follows from the fact that $\mathcal{R}=\bigsqcup_n \mathcal{R}_n$.
\end{proof}

If we combine Theorem \ref{thm:main} with Theorem \ref{thm:moments_F} and replace the indices $2k$ for $k\geq 0$ by $2g-2$ for $g\geq 1$, we can also somehow interpret the asymptotic expansions of the full partition function on the torus as a topological expansion. However, the actual coefficients are much more involved, and we are unable to express the partition function as an explicit integral over random ramified coverings as in the chiral case.

\begin{theorem}\label{thm:top_exp_torus}
For any $t>0,$  the partition function $Z_N^\YM(1,t)$ admits the following expansion for any $p\geq 1$:
\begin{equation}
Z_N^\YM(1,t) = \sum_{g=1}^{p+1} \frac{a_{2g-2}(t)}{N^{2g-2}} + O(N^{-2p-2)}).
\end{equation}
The coefficients $a_{2g-2}(t)$ are given by 
\begin{equation}
a_{2g-2}(t) = \sum \frac{t^{2g-2}(-1)^{k_2}D^{k_1}\F_{1,g_1}(q_t) D^{k_2}\F_{1,g_2}(q_t)D^{\frac{k_1+k_2}{2}}\theta(q_t)}{(2g_1-2)!(2g_2-2)!k_1!k_2!},
\end{equation}
where the sum is over $g_1,g_2,\geq 1,$ $k_1,k_2\geq 0$ such that $2g_1-2+2g_2-2+k_1+k_2=2g-2$.
\end{theorem}

\begin{remark}
The leading term in the large $N$ limit is given by $g=1$, which corresponds to $g_1=g_2=1$ and $k_1=k_2=0$, so that
\[
a_0(t) = \F_{1,1}(q_t)^2\theta(q_t)= \frac{\theta(q_t)}{\phi(q_t)^2}.
\]
The limit \eqref{eq:lim_pf_TL} is again recovered. The next order term, corresponding to $g=2$, is already much more involved, and is obtained by adding the terms corresponding to the following values of $(g_1,g_2,k_1,k_2)$: $(2,1,0,0)$, $(1,2,0,0)$, $(1,1,2,0)$, $(1,1,0,2)$ and $(1,1,1,1)$. We get
\begin{align*}
a_2(t) = & t^2\F_{1,2}(q_t)\F_{1,1}(q_t)\theta(q_t)+t^2\F_{1,1}(q_t)D^2\F_{1,1}(q_t)D\theta(q_t)\\
& - t^2\left(D\F_{1,1}(q_t)\right)^2 D\theta(q_t).
\end{align*}
There might be an algorithmic procedure for iteratively compute the coefficients $a_{2g-2}$ for higher $g$, using topological recursion techniques in the spirit of \cite{LMS}.
\end{remark}

\begin{appendix}
\section*{Dominated convergence on countable measured spaces}

\begin{lemma}\label{lem:cv_L1}
Let $X$ be a countable set, endowed with the $\sigma$-algebra $\mathcal{P}(X)$ of all its subsets, and let $\mu:\mathcal{P}(X)\to[0,\infty]$ be its counting measure. Let $(f_n)_{n\geq 1}$ be a sequence of integrable functions on $X$ that converges pointwise to 0 on $X$. Then, it converges to 0 in $L^1(\mu)$ if and only if for any $\varepsilon>0$, there exists a measurable subset $X_\varepsilon$ such that $\mu(X_\varepsilon)<\infty$ and
\begin{equation}\label{eq:Bogachev}
\sup_n \int_{X\setminus X_\varepsilon} \vert f_n\vert d\mu < \varepsilon.
\end{equation}
\end{lemma}

\begin{proof}
According to \cite[Corollary 4.5.5]{Bog}, which applies on any measured space with a positive measure, the convergence of $(f_n)$ in $L^1(\mu)$ holds if and only if \eqref{eq:Bogachev} holds true and
\[
\lim_{\mu(E)\to 0} \sup_n \int_E \vert f_n\vert d\mu = 0.
\]
If we prove that the latter holds for any countable measured space with its counting measure, then the lemma will follow.

For any $\varepsilon>0$, for any $E\subset X$ such that $\mu(E)<1$, we have $E=\varnothing$ and it is clear that
\[
\sup_n \int_{\varnothing} \vert f_n\vert d\mu = 0 < \varepsilon.
\]
Hence, the condition is satisfied.
\end{proof}
\end{appendix}
\bibliographystyle{imsart-number} 
\bibliography{main-aop}       


\end{document}